\documentclass[letterpaper, 9 pt, conference]{ieeeconf}      

\IEEEoverridecommandlockouts                              
\pdfoutput=1                                                        

\usepackage{graphics} 
\usepackage{epsfig} 
\usepackage{times} 
\usepackage{amsmath} 
\usepackage{amssymb}  
\usepackage{subfigure}
\usepackage{url}
\usepackage{dsfont}
\usepackage{mathtools}
\newtheorem{theorem}{Theorem}

\usepackage{cite}
\newtheorem{assumption}[theorem]{Assumption}
\usepackage{color}
\usepackage{algorithm}
\newtheorem{properties}[theorem]{Properties}

\newtheorem{definition}[theorem]{Definition}

\newtheorem{lemma}[theorem]{Lemma}

\newtheorem{proposition}[theorem]{Proposition}

{ 
\newtheorem{remark}[theorem]{Remark}

}

\newcommand{\bi}{\begin{itemize}}
\newcommand{\ei}{\end{itemize}}
\newcommand{\bd}{\begin{displaymath}}
\newcommand{\ed}{\end{displaymath}}
\newcommand{\be}{\begin{eqnarray*}}
\newcommand{\ee}{\end{eqnarray*}}

\newcommand{\T}{{\cal T}}
\usepackage[normalem]{ulem}

\title{\LARGE \bf
Information Transfer in Dynamical Systems and Optimal Placement of Actuators and Sensors for Control of Non-equilibrium Dynamics}
\author{Subhrajit Sinha, Umesh Vaidya and Enoch Yeung\\
\thanks{S. Sinha is with Pacific Northwest National Laboratory, Richland, Washington. U. Vaidya is with the Department of Electrical Engineering at Iowa State University, Ames, Iowa. E. Yeung is with Department of Mechanical Engineering, University of California, Santa Barbara, California.  \tt \small {email : subhrajit.sinha@pnnl.gov}
}
}

\begin{document}
\maketitle

\begin{abstract}
In this paper we develop the concept of information transfer between the Borel-measurable sets for a dynamical system described by a measurable space and a non-singular transformation. The concept is based on how Shannon entropy is transferred between the measurable sets, as the dynamical system evolves. We show that the proposed definition of information transfer satisfies the usual notions of information transfer and causality, namely, zero transfer and transfer asymmetry. Furthermore, we show how the information transfer measure can be used to classify ergodicity and mixing. We also develop the computational methods for information transfer computation and apply the framework for optimal placements of actuators and sensors for control of non-equilibrium dynamics.
\end{abstract}

\section{Introduction}
A mathematical study of dynamical system has a rich history, beginning with the work of Newton. A dynamical system is usually defined on a manifold $M$ with a family of smooth functions $\phi_t:M\to M$, where $t$ is time, generally belonging to a monoid $\T$ ($\mathbb{R}_{\geq 0}$ or $\mathbb{Z}_{\geq 0}$). The maps $\phi_t$ are such that they preserve the monoid structure, that is $\phi_t\circ \phi_s=\phi_{t+s}$ \cite{Arnold_classical_mechanics}. In another approach, a dynamical system is defined on a measure space and time evolution is modelled through a collection of measurable functions which preserve the monoid structure of $\T$ \cite{Lasota, Katok}.
In recent times, operator theoretic ideas are increasingly being used for analysis of nonlinear dynamical systems \cite{Dellnitz_Junge,Mezic2000,froyland_extracting,Junge_Osinga,Mezic_comparison,
Dellnitztransport,mezic2005spectral,Mehta_comparsion_cdc,Vaidya_TAC,
raghunathan2014optimal,susuki2011nonlinear,mezic_koopmanism,
mezic_koopman_stability,surana_observer,yeung2018koopman,yeung2017learning}. The main advantage of this approach is the fact that these operators, namely Perron-Frobenius (P-F) operator and Koopman operators are linear operators even though the underlying system may be nonlinear. Moreover, the operators are positive Markov operators and these properties can be exploited to have probabilistic interpretations.

On the other hand, the concept of entropy was first introduced by Clausius and later in the $19^{th}$ century Ludwig Boltzmann interpreted entropy as a measure of disorder. In particular, Boltzmann entropy is given by the famous relation 
\[S=k_B\ln W\]
where $S$ is entropy, $k_B$ is the Boltzmann constant and W is the number of microstates associated with a given macrostate of a system. Claude Shannon generalized this notion notion of entropy when he provided his information-theoretic version of entropy \cite{Shannon} and defined entropy for a probability density. Shannon's information theory is symmetric and in \cite{Marko} Shannon's information theory was generalized to give a sense of direction when the author defined bi-directional information. The idea was taken forward by Massey and Kramer \cite{IT_kramer_directedit, IT_massey_directed} and defined directed information. 

In dynamical system setting Liang and Kleeman introduced the concept of information transfer \cite{liang_kleeman_prl} which capture zero transfer and is inherently asymmetric. However, Liang-Kleeman information transfer is not designed to capture indirect influence. This was addressed in \cite{sinha_IT_CDC} and in \cite{sinha_IT_CDC,sinha_IT_ICC} the authors provided a new definition of information transfer between the state(s) of a dynamical system and it was shown that the information transfer thus defined captures the intuitions of causality, namely zero transfer, transfer asymmetry and also clearly distinguishes between direct and indirect transfers. Furthermore, the utility of the information transfer measure for analysis of nonlinear systems was demonstrated in \cite{IT_influence_acc,sinha_IT_power_CDC2017,sinha_power_journal}, where the information measure was used for stability classification and causality analysis of power networks. Moreover, since the theory is formulated in terms of the P-F operator, it allows the data-driven computation of information transfer measure for a dynamical system \cite{sinha_IT_data_acc,sinha_IT_data_journal}.

In this paper, we extend the definition of information transfer to define information transfers between the measurable sets of a dynamical system. In particular, we consider a dynamical system as a quadruple $(X,{\cal B},\mu,T)$, where $X$ is the phase space, $\cal B$ is the Borel $\sigma$-algebra on $X$, $\mu$ a probability measure and $T:X\to X$ and define the information transfer between the sets $A_i\in {\cal B}$, with respect to the measure $\mu$ and the map $T$. A direct consequence of the definition of information content of any set in the $\sigma$-algebra is the recovery of the famous Liuoville's Theorem \cite{goldstein_classical_mechanics} for measure-preserving systems. As with the definition of information transfer between the states \cite{sinha_IT_CDC, sinha_IT_ICC}, the information transfer defined here is further generalized to define information transfer over multiple iterates of the map $T$ and  with this we provide necessary and sufficient conditions under which the dynamical system is either ergodic or mixing. We also discuss the applicability of the developed framework and provide an information transfer based solution to the problem of optimal placement of actuators and sensors for control of non-equilibrium dynamics \cite{Vaidya_gramian_journal,vaidya_optimalplacement_ecc,sinha_optimal_placement}.

The paper is organized as follows. In sections \ref{section_preliminaries} and \ref{section_IT_states} we briefly review the concepts of transfer operators and information transfer between the states in a dynamical system. In section \ref{section_IT_phase} we define the information transfer measure between the Borel measurable sets and provide the main theoretical results followed by a discussion of computational technique of information transfer in section \ref{section_finite}. In section \ref{section_optimal} we address the problem of optimal placement of actuators and sensors and provide an information transfer based solution. Finally we conclude the paper in section \ref{section_conclusion}.

\section{Preliminaries}\label{section_preliminaries}
Consider a dynamical system
\begin{eqnarray}\label{system}
x_{n+1} = T(x_n)
\end{eqnarray}
where $T:X\to X$ is in general assumed to be at least differentiable and non-singular, with $X\in\mathbb{R}^N$. The mapping $T$ is said to be non-singular with respect to a measure $m$, if $m(T^{-1}A)=0$ for all $A\in\mathcal{B}(X)$, such that $m(A)=0$. Here $\mathcal{B}(X)$ denotes the Borel $\sigma$-algebra on $X$ and let $\mathcal{M}(X)$ be the vector space of real-valued measures on ${\cal B}(X)$. 

\subsection{Transfer Operators}
Associated with the dynamical system (\ref{system}) are transfer operators, namely Perron-Frobenius (P-F) operator and Koopman operator, which are used to study the dynamical system. Instead of studying the trajectories on the state space, these operators studies the evolution of functions under the mapping $T$. Before defining the P-F and Koopman operators, we define the more general Markov operator \cite{Lasota}.
\begin{definition}[Markov Operator]
Let $(X,{\cal A}, \mu)$ be a measure space. Any linear operator $P: {\cal L}_1\to {\cal L}_1$ satisfying
\begin{itemize}
\item{$Pf\geq 0$ for $f\geq 0$, $f\in{\cal L}_1$; and}
\item{$\parallel Pf \parallel = \parallel f \parallel$, for $f\geq 0$, $f\in{\cal L}_1$}
\end{itemize}
is called a Markov operator.
\end{definition}

Peron-Frobenius operator is a particular type of Markov operator and is used for the evolution of densities. It is defined as follows:

\begin{definition}[Perron-Frobenuis Operator]
Let $(X,{\cal A}, \mu)$ be a measure space. The Perron-Frobenius operator ${\mathbb{P}}:{\cal L}_1 \to {\cal L}_1$ corresponding to the dynamical system (\ref{system}) is a Markov operator satisfying
\begin{eqnarray}
\int_A \mathbb{P} f(x) d\mu = \int_{T^{-1}(A)}f(x)d\mu
\end{eqnarray}
\end{definition}
For invertible nonsingular transformations $T$, we have
\[[\mathbb{P}f](x)=f(T^{-1}(x))|\frac{dT^{-1}(x)}{dx}|,\]
where $|\cdot |$ is the determinant.

Another operator of interest is the Koopman operator $\mathbb{U} : {\cal L}_{\infty}\to {\cal L}_{\infty}$, which is defined as follows:

\begin{definition} [Koopman Operator] 
Given any $h\in{\cal L}_{\infty}$, $\mathbb{U}:{\cal L}_{\infty}\to {\cal L}_{\infty}$ is defined by
\[[\mathbb{U} h](z)=h(T(z)).\]
\end{definition}

Both the P-F operator and the Koopman operator are linear operators, even if the underlying system is non-linear. But while analysis is made tractable by linearity, the trade-off is that these operators are typically infinite dimensional. In particular, the P-F operator and Koopman operator often will lift a dynamical system from a finite-dimensional space to generate an infinite dimensional linear system in infinite dimensions. 
\begin{properties}\label{property}
Following properties for the Koopman and Perron-Frobenius operators can be stated \cite{Lasota}.

\begin{enumerate}
\item [a).] For the Hilbert space ${\cal F}=L_2(Z,{\cal B}, \bar \mu)$ 
\begin{eqnarray*}
&&\parallel \mathbb{U}h\parallel^2=\int_Z |h(T(z))|^2d\bar \mu(z)
\nonumber\\&=&\int_Z | h(z)|^2 d\bar\mu(z)=\parallel h\parallel^2
\end{eqnarray*}
where $\bar \mu$ is an invariant measure. This implies that Koopman operator is unitary.

\item [b).] For any $h\geq 0$, $[\mathbb{U}h](z)\geq 0$ and hence Koopman is a positive operator.

\item [c).]For invertible system $T$, the P-F operator for the inverse system $T^{-1}:Z\to Z$ is given by $\mathbb{P}^*$ and $\mathbb{P}^*\mathbb{P}=\mathbb{P}\mathbb{P}^*=I$. Hence, the P-F operator is unitary.

\item [d).] If the P-F operator is defined to act on the space of densities i.e., $L_1(Z)$ and Koopman operator on space of $L_\infty(Z)$ functions, then it can be shown that the P-F and Koopman operators are dual to each other \footnote{with some abuse of notation we use the same notation for the P-F operator defined on the space of measure and densities.}
\begin{eqnarray*}
&&\left<\mathbb{U} f,g\right>=\int_Z [\mathbb{U} f](z)g(z)dx\nonumber\\&=&\int_Xf(y)g(T^{-1}(y))\left|\frac{dT^{-1}}{dy}\right|dy=\left<f,\mathbb{P} g\right>
\end{eqnarray*}
where $f\in L_{\infty}(Z)$ and $g\in L_1(Z)$ and the P-F operator on the space of densities $L_1(Z)$ is defined as follows
\[[\mathbb{P}g](z)=g(T^{-1}(z))|\frac{dT^{-1}(z)}{dz}|.\]

\item [e).] For $g(z)\geq 0$, $[\mathbb{P}g](z)\geq 0$.

\item [f).] Let $(Z,{\cal B},\mu)$ be the measure space where $\mu$ is a positive but not necessarily the invariant measure of $T:Z\to Z$, then the P-F operator $\mathbb{P}:L_1(Z,{\cal B},\mu)\to L_1(Z,{\cal B},\mu)$  satisfies  following property:

 \[\int_Z [\mathbb{P}g](z)d\mu(z)=\int_Z g(z)d\mu(x).\]\label{Markov_property}
\end{enumerate}
\end{properties}

\section{Information Transfer Between States of a Dynamical System}\label{section_IT_states}
Information flow between the states of dynamical systems is a newly developed concept \cite{sinha_IT_CDC,sinha_IT_ICC} which studies how the evolution of a state(s) of a dynamical system affect/influence any other state(s). For completeness of the paper we briefly discuss the concept of information transfer between the states of a dynamical system. 

Consider the discrete time dynamical system
\begin{eqnarray}\label{system1}
z(t+1) = T(z(t)) + \xi(t)
\end{eqnarray}
where $z = \begin{pmatrix}
x^\top & y^\top
\end{pmatrix}^\top \in \mathbb{R}^N$, $T:\mathbb{R}^N\to \mathbb{R}^N$ is assumed to be at least continuous and $\xi(t)$ is independent and identically distributed additive noise, which comes from the distribution $g$. 

Information transfer from state (subspace) $x$ to state (subspace) $y$ gives a measure of how the evolution of $x$ dynamics affect (influence) the evolution of $y$ dynamics. In particular, we quantify this influence in terms of the entropy transferred from the $x$ dynamics to the $y$ dynamics, as the system (\ref{system1}) evolves in time. Note that, by entropy we mean the Shannon entropy.

The entropy of a distribution is the measure of the information content of the distribution. Suppose there are two agents (states in the case of a dynamical system) which are interacting with each other. Each has its own entropy (information) and they \emph{transfer} a part of their own information to the other agent via the interaction. We use this intuition to define the information transfer. In particular, information transfer from $x$ to $y$ is the amount of information (entropy) of $x$ that is being transferred to $y$, the system (\ref{system1}) evolves in time. With this, we define the information transfer as follows.

\begin{definition}[Information Transfer]
The information transfer from a state $x$ to state $y$ $({\cal T}_{x\to y})$, as the dynamical system $z(t+1) = T(z(t)) + \xi(t)$
evolves from time step $t$ to time step $t+1$ is defined as
{\small
\begin{eqnarray}\label{IT_def}
{\cal T}_{x\to y} = H(y(t+1)|y(t))-H_{\not{x}}(y(t+1)|y(t))
\end{eqnarray}}
where $z = \begin{pmatrix}
x & y
\end{pmatrix}^\top$, $H(y(t+1)|y(t))$ is the entropy of $y$ at time $t+1$ conditioned on $y(t)$ and $H_{\not{x}}(y(t+1)|y(t))$ is the conditional entropy of $y(t+1)$, conditioned on $y(t)$, when $x$ is absent from the dynamics.
\end{definition}
The intuition behind the definition of information transfer is the fact that the total entropy of $y$ is the entropy of $y$ when $x$ is absent from the dynamics plus the entropy transferred from $x$ to $y$. Hence, the information transfer from $x$ to $y$ gives the amount of entropy flowing from $x$ to $y$ and thus quantifies how much the $x$ dynamics affects the $y$ dynamics. With this, we define influence in a dynamical system as follows,
\begin{definition}[Influence]\label{def_influence}
We say a state (or subspace) $x$ influences a state (or subspace) $y$ if and only if the information transfer from $x$ to $y$ is non-zero.
\end{definition}
Larger the absolute value of the $T_{x\to y}$, more is the effect of $x$ dynamics on $y$ and hence more is the influence of $x$ on $y$. 

The information transfer thus defined (\ref{IT_def}), gives the entropy transferred as the system (\ref{system1}) evolves by one time step. This gives whether the state (or subspace) $x$ affect/influence the state (or subspace) $y$ directly. The definition of information transfer can be generalized to define information transfer over multiple time steps \cite{sinha_IT_CDC} and it has been shown to capture indirect influence in a dynamical system \cite{sinha_IT_CDC}.

\section{Information Transfer Between Sets in the Phase Space}\label{section_IT_phase}
In this section, we develop the notion of information transfer between the sets of a dynamical system. 

\subsection{Shannon Entropy and Dynamical Systems}
Let $(X,{\cal B},\mu)$ be a measure space, such that $\mu(X)<\infty$. Hence, by normalizing with $\mu(X)$ one can define a probability measure on $(X,{\cal B})$ and treat $(X,{\cal B},\mu)$ as a probability space. Let $T:X\to X$ be a transformation on $X$ and $\rho$ be a probability density on $X$, such that $\int_X \rho d\mu = 1$. Then the Shannon entropy of the density $\rho$ is given by
\[H(\rho) = - \int_X\rho\log\rho d\mu.\]
The evolution of $\rho$, as the dynamical system $(X,{\cal B}, \mu, T)$ evolves under the transformation $T$ is given by the Perron-Frobenius operator, which is defined as
\begin{definition}[Perron-Frobenius Operator]\cite{Lasota}
The P-F operator $\mathbb{P}:{\cal M}(X)\to {\cal M}(X)$ is given by

{\small
\begin{eqnarray}
[\mathbb{P}\mu](A)=\int_{{\cal X} }\delta_{T(x)}(A)d\mu(x)=\mu(T^{-1}(A))
\end{eqnarray}
}
where $\delta_{T(z)}(A)$ is stochastic transition function which measure the probability that point $x$ will reach the set $A$ in one time step under the system mapping $T$ and ${\cal M}(X)$ is the space of signed measures on $X$. 
\end{definition}
With this we have
\begin{theorem}\label{theorem_entropy_increase}
Let $(X,{\cal B},\mu)$ be a finite measure space, such that $\mu(X)<\infty$ and $P:{\cal L}_1\to {\cal L}_1$ be a Markov operator. If $P$ has a constant stationary density $(P1 = 1)$, then
\[H(P\rho)\geq H(\rho).\]
\end{theorem}
\begin{proof}
For any $f\in{\cal L}_1$, such that $f\geq 0$, define 
\[\eta(f) = - f\log f.\]
Then, $\eta$ is a concave function and hence from Jensen inequality we have
\[\eta (P\rho)\geq P\eta(\rho).\]
Integrating over the entire space $X$, we have
\begin{eqnarray}\nonumber
\int_X \eta (P\rho(x))d\mu &&\geq \int_X P\eta(\rho(x))d\mu\\
&&=\int_X \eta(\rho(x))d\mu
\end{eqnarray}
since $P$ preserves the integral. Hence, we have
\[H(P\rho)\geq H(\rho).\]
\end{proof}

Note that, for a finite measure space, when $P$ has a constant stationary density, the above theorem tells us that the entropy never decreases. Note that, this statement is similar to the second law of thermodynamics.

The above theorem states that under the transformation $T$, when the Markov operator has a constant stationary density, entropy never decreases. However, the entropy may not increase at all during iterations of $T$. This is the case when the Markov operator $P$ is the Perron-Frobenius operator $(\mathbb{P})$ and $T$ is an invertible measure preserving transformation.

\begin{definition}[Measure Preserving Transformation]
Let $(X,{\cal B},\mu)$ be a measure space and $T:X\to X$ a measurable transformation. Then $T$ is said to be measure-preserving if
\[\mu(T^{-1}(A))=\mu(A),\quad \forall A\in{\cal B}.\]
\end{definition}

\begin{theorem}
Let $(X,{\cal B},\mu)$ be a finite measure space and $T:X\to X$ be an invertible measure-preserving transformation. If $\mathbb{P}$ is the Perron-Frobenius operator corresponding to $T$, then $H(\mathbb{P}^nf)=H(f)$ for all $n$, where $f\in{\cal L}_1$ and $f\geq 0$.
\end{theorem}
\begin{proof}
For $f\in{\cal L}_1$, we have
\[[\mathbb{P}f](x)=f(T^{-1}(x))|\frac{dT^{-1}(x)}{dx}|.\]
Since, $T$ is invertible and measure-preserving, $|\frac{dT^{-1}(x)}{dx}|=1$. Hence,
\[[\mathbb{P}f](x)=f(T^{-1}(x)).\]
Let $\mathbb{P}_1$ be the P-F operator corresponding to $T^{-1}$. Then by similar arguments, $[\mathbb{P}_1f](x)=f(T(x))$. Hence, $\mathbb{P}_1\mathbb{P}f = \mathbb{P}\mathbb{P}_1f$. Thus, 
\[\mathbb{P}_1=\mathbb{P}^{-1}.\]
Hence, from Theorem \ref{theorem_entropy_increase},
\[H(\mathbb{P}_1\mathbb{P}f)\geq H(\mathbb{P}f)\geq H(f).\]
However, since $\mathbb{P}_1\mathbb{P}f = f$, we have $H(\mathbb{P}f)=H(f)$. Hence, 
\[H(\mathbb{P}^nf)=H(f), \quad \forall n.\]
\end{proof}

Hence, for invertible measure-preserving dynamical systems, the entropy remains constant. An example is the Hamiltonian systems. However, for Hamiltonian systems, one has the celebrated Liouville's Theorem \cite{goldstein_classical_mechanics} which proves that the volume of any set in the phase space is conserved. In order to state the theorem in terms of entropy, it is not sufficient to consider the entropy over the entire phase space $X$, but one has to define the entropy of any measurable set $A\subset {\cal B}$. 

This assignment of Shannon entropy to any measurable subset of the phase space $X$ will form the basis of the definition of information flow between the sets in the phase space of a dynamical system and this will be the main focus in the next subsection.

\subsection{Entropy and Information Flow}
In this subsection we define the information transfer between the sets in the phase space of a dynamical system. Let $(X,{\cal B}, \mu)$ be a probability space. Then the Shannon entropy of any measurable set can be defined as follows:

\begin{definition}
The Shannon entropy of any set $A\in{\cal B}$ with respect to the measure $\mu$ is defined as
\[H_{\mu}(A)= - \mu(A)\log \mu(A).\]
If $\mu(A)$ is zero, then $H_{\mu}(A)\coloneqq 0$.
\end{definition}

\begin{lemma}
Let $(X,{\cal B},\mu)$ be a probability space and $T:X\to X$ be an invertible measure-preserving transformation. Let $A\in {\cal B}$, then the Shannon entropy of $A$ remains constant under the action of the transformation $T$.
\end{lemma}
\begin{proof}
Since $T$ is an invertible measure-preserving transformation on $X$, we have
\[\mu(T^n(A)) = \mu(A), \forall n\in\mathbb{Z}.\]
Hence for $n\in\mathbb{Z}$,
\begin{eqnarray*}
H_{\mu}(T^n(A)) &=& -\mu(T^n(A))\log\mu(T^n(A))\\
&=& -\mu(A)\log\mu(A) = H_{\mu}(A).
\end{eqnarray*}
Hence the entropy of any measurable set remains constant under the transformation $T$.
\end{proof}
The above lemma gives a conservation law for the Shannon entropy and in the framework of Hamiltonian systems is the Liouville's Theorem \cite{goldstein_classical_mechanics}, which states that phase space volumes are conserved under a Hamiltonian flow.

With the Shannon entropy defined for any measurable set $A\in {\cal B}$, with respect to a measure $\mu$, now we are in a position to define the information flow among any two sets measurable sets $A_i,A_j\subset {\cal B}$.

\begin{figure}[htp!]
\centering
\includegraphics[scale=.175]{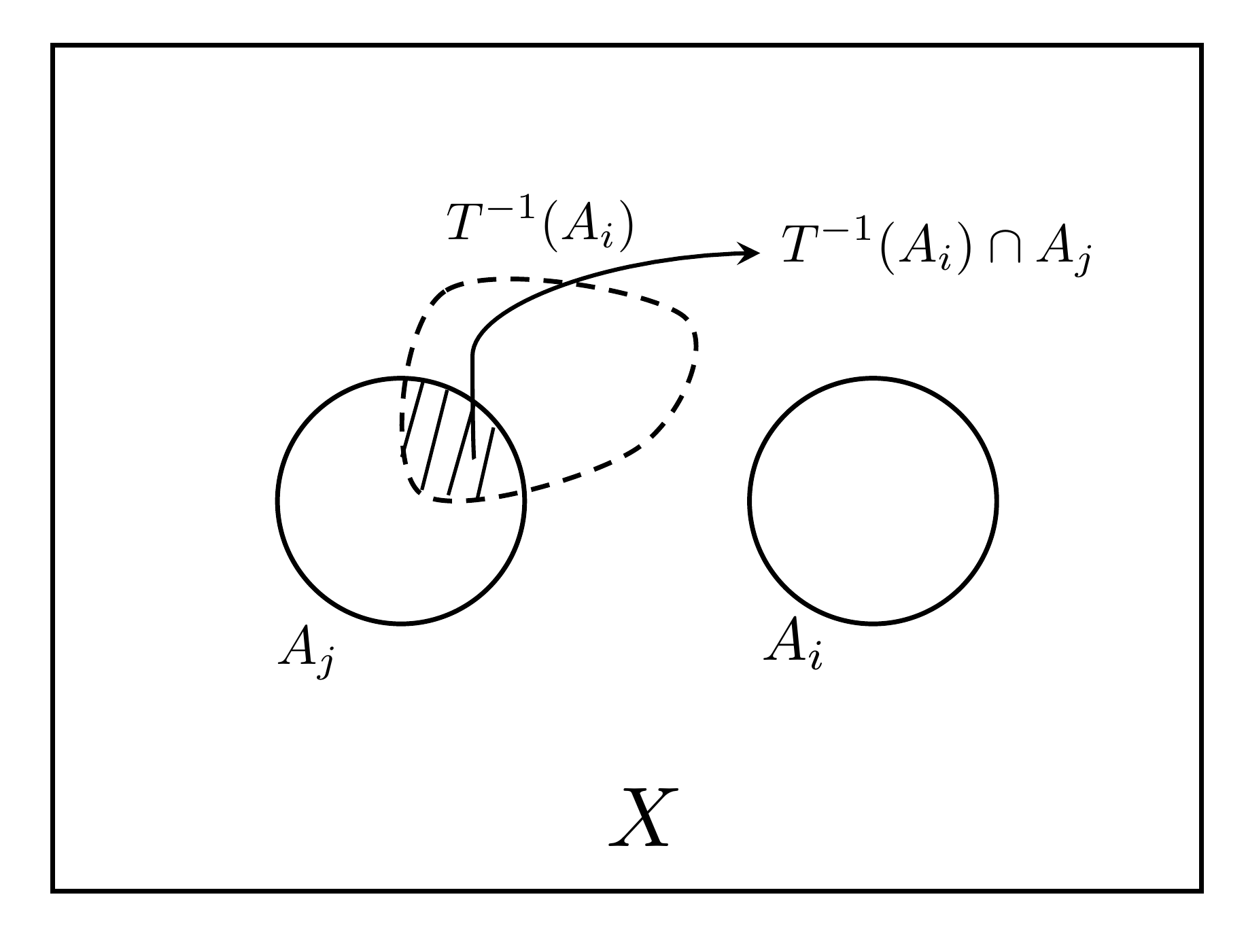}
\caption{Schematic showing intuition of information transfer.}\label{fig_IT_intuition}
\end{figure}

\begin{definition}[Information Transfer]\label{def_IT_set}
Let $(X,{\cal B}, \mu)$ be a probability space and let $T:X\to X$ be a non-singular transformation. Let $A_i,A_j\subset {\cal B}$ such that $\mu(A_i)>0$. Define $\mu_{ij}$ as
\begin{eqnarray}\label{mu_ij}
\mu_{ij} = \frac{\mu(T^{-1}(A_j)\cap A_i)}{\mu(A_i)}.
\end{eqnarray}
Then the information transferred from set $A_i$ to $A_j$, with respect to the measure $\mu$, under one iterate of the transformation $T$ is
\begin{eqnarray}\label{def_IT_set_expression}
{\cal T}_{i\to j}^{\mu} = -\mu_{ij}\log\mu_{ij}.
\end{eqnarray}
If $\mu(A_i)=0$, then ${\cal T}_{i\to j}^{\mu}\coloneqq 0$.
\end{definition}
For invertible transformations $T$, the quantity $\mu_{ij}$ can be equivalently defined as
\[\mu_{ij} = \frac{\mu(T(A_i)\cap A_j)}{\mu(A_i)}.\]
Intuitively, the quantity $\mu_{ij}$ quantifies the fraction of the volume of $A_i$ that ends up in the set $A_j$ under one iterate of the transformation $T$ (see Fig. \ref{fig_IT_intuition}) and information transfer, thus defined in definition \ref{def_IT_set}, gives the information flow from $A_i$ to $A_j$ under one iterate of the transformation $T$. Hence this can be thought of as one-step information transfer. 
\begin{proposition}
Let $(X,{\cal B},\mu)$ be a probability space and $T:X\to X$ a non-singular transformation. Then the information transfer defined in (\ref{def_IT_set_expression}) satisfies the following:
\begin{enumerate}
    \item {${\cal T}_{i\to j}^{\mu}\geq 0.$ for $A_i,A_j\in {\cal B}$ such that $\mu(A_i)>0$ and $\mu(A_j)>0$. Moreover, if $\mu(A_j)$ is zero, then the information transfers ${\cal T}_{i\to j}^{\mu}$ is zero.}
    \item{Information transfer, in general, is asymmetric.}
\end{enumerate}
\end{proposition}
\begin{proof}
1) Consider $A_i,A_j\in {\cal B}$ such that $\mu(A_i)>0$ and $\mu(A_j)>0$. Then the non-negativity of the information transfer follows directly from the definition of Shannon entropy and ${\cal T}_{i\to j}^{\mu}= 0$ if and only if $\mu(T^{-1}(A_j)\cap A_i)=0.$ Moreover, let $\mu(A_j)=0$. Since $T$ is non-singular $\mu(T^{-1}(A_j))=0$. And since $0\leq\mu(A)\leq 1$ for $A\in {\cal B}$ and $\mu(T^{-1}(A_j)\cap A_i)\leq \mu(T^{-1}(A_j))$, we have \[\mu(T^{-1}(A_j)\cap A_i)=0\]
and hence ${\cal T}_{i\to j}^{\mu}= 0$.

2) In a system, it may happen that $T^{-1}(A_j)\cap A_i$ is a measure zero set, but $T^{-1}(A_i)\cap A_j = A_k$, such that $\mu(A_k)>0$. In this case ${\cal T}_{i\to j}^{\mu}= 0$ but ${\cal T}_{j\to i}^{\mu}\geq 0$. Hence, information transfer in general is not symmetric.
\end{proof}

Unlike mutual information between two measurable sets, which is symmetric and quantifies the information shared by two sets, it is the asymmetry property of information transfer which allows us to characterize influence of one measurable set on any other measurable set. In \cite{sinha_IT_CDC}, the one-step information transfer between the states in a dynamical system was generalized to define information transfer over multiple time steps. In this set-theoretic setting this amounts to defining information transfer from $A_i$ to $A_j$ under the transformation $T^n$, for $n\in\mathbb{N}$.

\begin{definition}[n-step Information Transfer]\label{def_IT_set_n_step}
Let $(X,{\cal B}, \mu)$ be a probability space and let $T:X\to X$ be a non-singular transformation. Let $A_i,A_j\subset {\cal B}$. Define $\mu_{ij}^n$ as
\[\mu_{ij}^n = \frac{\mu(T^{-n}(A_j)\cap A_i)}{\mu(A_i)}.\]
Then the information transferred from set $A_i$ to $A_j$, with respect to the measure $\mu$, under $n$ iterates of the transformation $T$ is
\begin{eqnarray}\label{def_IT_set_expression_n_step}
[{\cal T}_{i\to j}^{\mu}]^n = -\mu_{ij}^n\log\mu_{ij}^n.
\end{eqnarray}
\end{definition}

Hence, the total information transferred from $A_i$ to $A_j$ as the dynamical system $(X,{\cal B},\mu, T)$ evolves $n$ times is
\begin{eqnarray}\label{IT_total}
[\T_{ij}^\mu]_1^n = \sum_{i = 1}^n [\T_{i\to j}^\mu]^i.
\end{eqnarray}

\subsection{Ergodicity, Mixing and Information Flow}
The term ``\emph{ergodic}" was first used by Ludwig Boltzmann in relation to problems in statistical mechanics. 
\begin{definition}[Ergodicity \cite{Lasota}]
Let $(X,{\cal B},\mu)$ be a measure space and $T:X\to X$ a non-singular transformation. Then $T$ is called ergodic if every invariant set $A\in{\cal B}$ is such that $\mu(A)=0$ or $\mu(X\setminus A) = 0$; that is, $T$ is ergodic if all invariant sets are trivial subsets of $X$.
\end{definition}

The above is one of the many equivalent definitions of ergodicity and depending on the underlying system and the problem in hand a particular equivalent formulation of ergodicity is used. The equivalent definition of ergodicity relevant to the current study is the following:
\begin{definition}[Ergodicity]
Let $(X,{\cal B}, \mu)$ be a measure space and let $T:X\to X$ be a non-singular transformation. Then $T$ is ergodic if for any two measurable sets $A,B\in{\cal B}$ with $\mu(A)>0$ and $\mu(B)>0$, there is an $n\in\mathbb{Z}$ such that $\mu(T^{-n}(A)\cap B)>0$.
\end{definition}

\begin{theorem}
Let $(X,{\cal B},\mu)$ be a probability space and $T:X\to X$ be a non-singular measure-preserving transformation. Then $T$ is ergodic if and only if for any two sets $A,B\in{\cal B}$, with $\mu(A)>0$ and $\mu(B)>0$ the total information transferred (\ref{IT_total}) from $A$ to $B$ ($[\T_{AB}^\mu]_1^n$) is non-zero.
\end{theorem}
\begin{proof}
First assume that $T$ is ergodic. Hence, for any two sets $A,B\in{\cal B}$ with $\mu(A)>0$ and $\mu(B)>0$, there exists $n\in\mathbb{Z}$ such that 
\[\mu(T^{-n}(A)\cap B)>0.\]
Hence for that $n$, 
\[\mu_{AB}^n = \frac{\mu(T^{-n}(A)\cap B)}{\mu(A)}>0.\]
Hence, $[\T_{A\to B}^\mu]^n$ is non-zero and since the information transfer is always positive, the total information flow from $A$ to $B$, given by
\[[\T_{AB}^\mu]_1^n = \sum_{i = 1}^n [\T_{A\to B}^\mu]^i\]
is non-zero.

On the other hand, assume that for any two sets $A,B\in{\cal B}$, with $\mu(A)>0$ and $\mu(B)>0$ the total information transfer $[\T_{AB}^\mu]_1^n= \sum_{i = 1}^n [\T_{A\to B}^\mu]^i$ is non-zero. This implies that there exists at least one $n\in\mathbb{Z}$ such that $[\T_{A\to B}^\mu]^n$ is non-zero. Moreover, since each $[\T_{A\to B}^\mu]^i>0$, it implies there exists at least one $n$ such that $[\T_{A\to B}^\mu]^n>0$. For that $n$, since the information transfer is positive, $\mu_{AB}^n>0$. Hence,
\begin{eqnarray}\nonumber
\mu_{AB}^n &=& \frac{\mu(T^{-n}(A)\cap B)}{\mu(A)}>0\\
&\implies& \mu(T^{-n}(A)\cap B)>0.
\end{eqnarray}
Hence $T$ is ergodic.
\end{proof}

 Although a plethora of ergodic systems can be constructed or can be abstractly shown to exist, it is extremely difficult to verify ergodicity for naturally occurring systems. Hence, in many cases, ergodicity is proved by proving the system to be satisfying a ``stronger" condition called \emph{mixing}.

\begin{definition}[Mixing \cite{Lasota}]
 Let $(X,{\cal B},\mu)$ be a probability space and $T:X\to X$ a measure-preserving transformation. $T$ is called mixing if
\[\lim_{n\to \infty}\mu(T^{-n}(A)\cap B)=\mu(A)\mu(B),\forall A,B\in{\cal B}.\]
\end{definition}

As a direct consequence of the definition of mixing we have the following lemma which characterizes mixing using information transfer.

\begin{lemma}
Let $(X,{\cal B},\mu)$ be a probability space and $T:X\to X$ a non-singular transformation. Let $A,B\in{\cal B}$ such that $\mu(A)>0$ and $\mu(B)>0$. Then $T$ is mixing if and only if the $n$-step information transfer, with $n\to\infty$, from $A$ to $B$ is equal to the  entropy of $B$.
\end{lemma}
\begin{proof}
The proof is straight-forward and follows from the fact that
\[\mu_{AB}^n = \frac{\mu(T^{-n}(A)\cap B)}{\mu(A)}\]
and for $T$ to be mixing one should have
\[\lim_{n\to \infty}\frac{\mu(T^{-n}(A)\cap B)}{\mu(A)}=\mu(B),\forall A,B\in{\cal B}.\]
\end{proof}
As stated earlier, mixing is stronger condition than ergodicity and mixing implies ergodicity. This fact is can be understood intuitively using the concept of information transfer as follows:

For ergodicity, the necessary and sufficient condition is that there should be a non-zero information flow  from any set of positive measure to any other set of positive measure. It does not impose any constraint on the quantity of information transferred. However, for the transformation to be mixing, the information transfer should not only be non-zero but should also be equal to a fixed entropy, determined by the measure of the set to which information is flowing. Thus mixing is a stronger condition and mixing implies ergodicity. Hence, information transfer can be used to distinguish ergodic transformations from mixing.

\section{Finite-dimensional Approximation}\label{section_finite}
In this section, we develop the computational framework for computation of information transfer between the states in the phase space of a dynamical system. The computation is based on set-oriented methods for construction of finite dimensional approximation of the Perron-Frobenius operator associated with a dynamical system $(X,{\cal B},\mu , T)$ \cite{dellnitz2002set}. For construction of the finite-dimensional approximation of the P-F operator, consider a finite partition (${\cal X}$) of $X$ as 
\begin{eqnarray}\label{partition}
{\cal X} = \{D_1,D_2,\cdots , D_N\},
\end{eqnarray}
such that $\cup_i D_i = X, D_i\cap D_j = \{\phi\}$ for $i\neq j$. Then the finite-dimensional P-F operator on the partition $\cal X$ is a $N\times N$ matrix $P$ such that
\begin{eqnarray}\label{PF_finite}
[P]_{ij}= \frac{m(T^{-1}(D_j)\cap D_i)}{m(D_i)},
\end{eqnarray}
where $[P]_{ij}$ is the $ij^{th}$ entry of the finite-dimensional $P$ matrix and $m$ is the Lebesgue measure. Computationally several short-term trajectories are used to compute the matrix $P$. In particular, the map $T$ is used to propagate $L$ ``initial conditions", which are chosen to be uniformly distributed over each cell $D_i$ and the entry $[P]_{ij}$ is approximated by the fraction of the initial conditions in $D_i$ that end up in cell $D_j$ after one iterate of the map $T$.

With the finite-dimensional P-F operator, the computation of information transfer between any two sets $D_i,D_j\in {\cal X}$ is straight-forward. In particular, from (\ref{mu_ij}) and (\ref{PF_finite}) we have $\mu_{ij} = [P]_{ij}$. Hence the information transfer from cell $D_i$ to cell $D_j$, with respect to the partition $\cal X$, is 
\begin{eqnarray}\label{IT_finite}
[\T_{D_i\to D_j}]= -[P]_{ij}\log [P]_{ij}. 
\end{eqnarray}
The expression in (\ref{IT_finite}) is the one-step information transfer and higher order transfers can be obtained similarly by considering powers of the finite-dimensional $P$ matrix. In particular, the information transferred from cell $D_i$ to cell $D_j$ under $n$ iterates of the map $T$ is
\begin{eqnarray}\label{IT_n_step_finite}
[\T_{D_i\to D_j}]^n = -[P^n]_{ij}\log [P^n]_{ij}
\end{eqnarray}
where 
\begin{eqnarray}
[P^n]_{ij}= \frac{m(T^{-n}(D_j)\cap D_i)}{m(D_i)}
\end{eqnarray}
is the $ij^{th}$ entry of the matrix $P^n$.

\section{Optimal Placement of Actuators and Sensors}\label{section_optimal}
\subsection{Problem Formulation}
Motivated by the problem of control of oil spill in fluid flow, we consider the problem of optimal placement of static actuators and sensors. The spatial distribution of the oil spill is modelled as a scalar density function $\rho(x,t)\in {\cal L}_2(X)$ which is advected by the differential equation
\[\dot{x} = f(x)\]
where $x\in X\subset \mathbb{R}^N$ and $X$ is assumed compact. The evolution of the density $\rho(x,t)$ is given by the linear advection PDE
\begin{eqnarray}\label{advection_pde}
\frac{\partial \rho(x,t)}{\partial t}=\nabla \cdot (f(x)\rho(x,t)).
\end{eqnarray} 
Let $A_k\subset X$, $k = 1, \cdots , p$ and $S_l\subset X$, $l = 1, \cdots , q$ be the locations of the actuators and sensors.
\begin{assumption}
We assume $0< m(A_k)\ll m(X)$ and $0 < m(S_l) \ll m(X)$, where $m$ is the Lebesgue measure.
\end{assumption}
The controlled evolution of linear advection PDE (\ref{advection_pde}) with spatially located actuators and sensors can be described as follows:
{\small
\begin{eqnarray}
\begin{aligned}
&\frac{\partial \rho(x,t)}{\partial t}=-\nabla \cdot (f (x)\rho(x,t))+\sum_{k=1}^p \chi_{A_k}(x)u_k(x,t)\\
&y_\ell(x,t)= \chi_{S_\ell}(x)\rho(x,t),\;\;\;k=1,\ldots, p,\;\;\;\ell=1,\ldots,q
\label{control_pde}
\end{aligned}
\end{eqnarray}
}
 where $\chi_{A_k}(x)$ and $\chi_{S_\ell}(x)$ are the indicator functions for the set $A_k$ and $S_\ell$ respectively.  $u_k(x,t)$ is the control input for the $k^{th}$ actuator and $y_\ell$ is the output of the $\ell^{th}$ sensor.

The objective is to provide an approach for optimal locations of $A_k$ and $S_l$, such that certain controllability and observability conditions are satisfied for the controlled PDE (\ref{control_pde}).

\subsection{Controllability and Observability of Controlled PDE}
In this subsection, we briefly state the main results about controllability and observability of the controlled PDE (\ref{control_pde}). For details see \cite{Vaidya_gramian_journal,vaidya_optimalplacement_ecc,sinha_optimal_placement}. 

In \cite{vaidya_optimalplacement_ecc, sinha_optimal_placement}, the authors had provided results connecting the flow of the underlying vector field and infinite-time controllability and observability of the controlled PDE (\ref{control_pde}). The characterization is in terms of the P-F and the Koopman operators and hence for implementation of the results we consider the finite-dimensional approximation of these operators on a partition of the space $X$.
\begin{remark}
For computation of the finite-dimensional P-F operator, the continuous time system $\dot{x}=f(x)$ is discretized using Euler discretization.
\end{remark}


The finite dimensional approximation of the P-F operator leads to a  coarser notion of controllability which we refer to as {\it coarse controllability}. The definition of coarse controllability closely follows that of coarse stability as introduced in \cite{Vaidya_TAC}.

\begin{definition}[Coarse Controllability \cite{sinha_optimal_placement}]\label{definition_coarse}
Consider a space $X$ with a finite partition $\mathcal{X}$. $X$ is said to be coarsely controllable with respect to the partition $\mathcal{X}$ if for an uncontrollable subset $A\subset U \subset X$, there exists no subpartition $S = \{D_{s_1}, D_{s_2},\cdots, D_{s_l}\}$ in $\mathcal{X}$ with domain $S = \cup_{k=1}^lD_{s_k}$, such that $A\subset S$ and $S$ is an uncontrollable subset of $X$.
\end{definition}

\begin{definition}[Coarse Controllability from $D_k$ \cite{sinha_optimal_placement}]\label{def_coarse_controllable_Dk}
Let $\cal X$ be the partition of the space and $D_k\in \cal X$. The system is said to be coarse controllable from cell $D_k$ if all the cells of the partition $\cal X$ are reachable from $D_k$.
\end{definition}
For typical partitions, coarse controllability means controllability modulo the uncontrollable sets $A$, with the uncontrollable sets smaller than the size of the cells within the partitions. When the cell sizes (measure) goes to zero, one recovers the uncontrollable sets fully. Without loss of generality, we make following assumption on the location of the actuators and sensors.
\begin{assumption}\label{assumption_location}
\[A_{i}=D_{k_i},\;\;\; i=1,\ldots, p \;\;{\rm where}\;\;k_i\in N_a.\]
where, $N_a\subseteq \{1,\ldots,N\}$ is the admissible set and it consists of the admissible locations where the actuators can be placed. The assumption essentially implies that the actuator occupies the entire cell of the partition $\cal X$. A similar assumption can also be made on the location of the sensor.
\end{assumption}
With this, we have the following characterization of the controllable region in terms of the finite-dimensional P-F operator $P$.
\begin{theorem}\label{theorem_coarse_controllable}
Let $P$ be the finite-dimensional approximation of the P-F semigroup constructed on  the finite partition $\mathcal{X} = \{D_1,D_2,\cdots , D_N\}$ of $X$. Let $A_i=D_{k_i}$, $i = 1,2, \cdots , p$ correspond to the location of actuator cells. The following are true.
\begin{enumerate}
\item If the entire space $X$ is controllable with actuator located on cells $D_{k_i}$ for $i=1,\ldots, p$ i.e.,
    \[\mu_c(B)=\sum_{n=0}^\infty \alpha^n \left[\mathbb{P}^n \sum_{i=1}^p m_{D_{k_i}} \right](B)\neq 0 \]
    for every set $B\subset X$ with $m(B)>0$, then the vector $\bar \mu=(\bar \mu_1,\ldots,\bar \mu_N) \in \mathbb{R}^N$ is nonzero i.e., $\bar \mu_k\neq 0$ for $k=1,\ldots,N$, where
    \begin{eqnarray}\bar \mu=\sum_{n=0^\infty}\alpha^n \sum_{i=1}^p e_{k_i}^\top P^n\label{equation1}
    \end{eqnarray}
    and $e_{k_i}\in \mathbb{R}^N$ is a column vector consisting of all zeros with $1$ at $k_i^{th}$ location.
 \item If $\bar \mu$ as defined in Eq. (\ref{equation1}) is nonzero (element wise) then then system is coarse controllable.
   \end{enumerate}
\end{theorem}
\begin{proof}
See \cite{sinha_optimal_placement}.
\end{proof}
Essentially, this theorem proves that if the actuators are located in the cells $D_{k_i}$, then under the flow all the other cells are reachable from these actuator cells. Similar results exists for sensor placement \cite{sinha_optimal_placement} where if the sensors are locates in cells $D_i$, then these cells are reachable from all the other cells. For results on coarse observability, see \cite{sinha_optimal_placement}.

\subsection{Information Transfer and Optimal Placement}
Let ${\cal X} = \{D_1, \cdots , D_N\}$ be a finite partition of the space $X$. Then from (\ref{IT_n_step_finite}), it can be seen that the cell $D_j$ is reachable from cell $D_i$ if and only if there exists a $n$ such that the $n$-step information transfer is non-zero. Note that if the $N-1$-step information transfer from $D_i$ to $D_j$ is zero, then for all $n>N-1$ the $n$-step information transfer remains zero. This follows from the fact that the finite partition can be thought of as a graph of $N$ nodes and if one node $(D_j)$ is not reachable from some other node $(D_i)$ in steps less than or equal to $N-1$, then $D_j$ is not reachable from $D_i$ at all. See \cite{sinha_optimal_placement} for a detailed discussion on this. 

Hence, the problem of optimal placement of actuators reduces to finding the minimum number of cells in the partition, such that there is non-zero information transfer from these cells to all the other cells. The optimization problem can be formulated in many different ways and we discuss the very obvious formulation. For other formulations see \cite{vaidya_optimalplacement_ecc,sinha_optimal_placement}. Let $e\in\mathbb{R}^{1\times N}$ be a vector of length $N$ with entries zero and one. Then one can write the optimal placement problem as
\begin{eqnarray}\label{optimization_problem}
\begin{aligned}
&\qquad\textnormal{ min } \quad||e||_0\\
&\textnormal{subject to } \quad e\T > 0\\
& \quad\qquad\quad e_i\in\{0,1\}
\end{aligned}
\end{eqnarray}
where $\T$ is the information transfer matrix such that $[\T]_{ij}$ is the total information transfer from cell $D_i$ to cell $D_j$ in $N-1$ time steps. However, the optimization problem (\ref{optimization_problem}) is non-convex. A relaxed version of (\ref{optimization_problem}) is a convex problem and the optimization problem can be stated as
\begin{eqnarray}\label{optimization_convex}
\begin{aligned}
&\qquad\textnormal{ min } \quad||e||_1\\
&\textnormal{subject to } \quad e\T > 0\\
& \quad\qquad\quad e_i\in [0,1], \sum_i e_i = p,
\end{aligned}
\end{eqnarray}
where $p$ is the number of actuators. Note that the relaxed problem provides a sub-optimal solution and the optimal solution $e^{\star}$, in general, will not be binary and one can provide a sub-optimal solution to the actuator placement problem by considering the largest $p$ values of $e^{\star}$ and place the actuators in those cells. The relaxed optimization places the given number of actuators in those cells such that the non-zero information transfer region is maximized, thus maximizing the coarse controllable region.

For the sensor placement problem the optimization problem, the objective is to choose the minimum number of cells such that information transfer from each of the other cells to these cells is non-zero, thus achieving coarse observability. Hence the relaxed optimization problem is
\begin{eqnarray*}
&&\qquad\textnormal{ min } \quad||e||_1\\
&&\textnormal{subject to } \quad \T^\top e^\top > 0\\
&& \quad\qquad\quad e_i\in [0,1], \sum_i e_i = q
\end{eqnarray*}
where $q$ is the number of sensors. 
\subsection{Simulation Results}
\emph{1. Control of contaminants in Double Gyre flow field}
\newline
We consider the double gyre fluid flow field given by
\[\dot x=-\pi \sin(\pi x)\cos(\pi y),\;\;\;
\dot y =\pi \sin(\pi y)\cos(\pi x)\]
where $(x,y)\in X=[0,2]\times [0,1]$ (Fig. \ref{fig_doublegyre}). The Double Gyre flow field is often used as a model for oceanographic flow field \cite{Wiggins_mixing_ocean} and the interest in this example comes from the problem of control of contaminants in oceanographic flows \cite{Igor_oilspill,haller_oilspill}. 
\begin{figure}[htp!]
\centering
{\includegraphics[scale=.325]{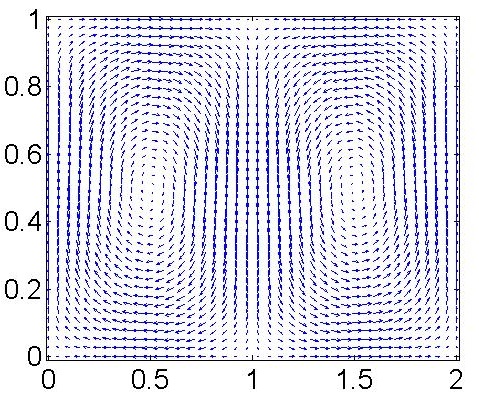}}
\caption{Double gyre velocity field }\label{fig_doublegyre}
\end{figure}

For the finite-dimensional approximation we consider a $60\times 60$ partition of the state space and we choose a total of 20 actuators.

\begin{figure}[htp!]
\centering
\subfigure[]{\includegraphics[scale=.225]{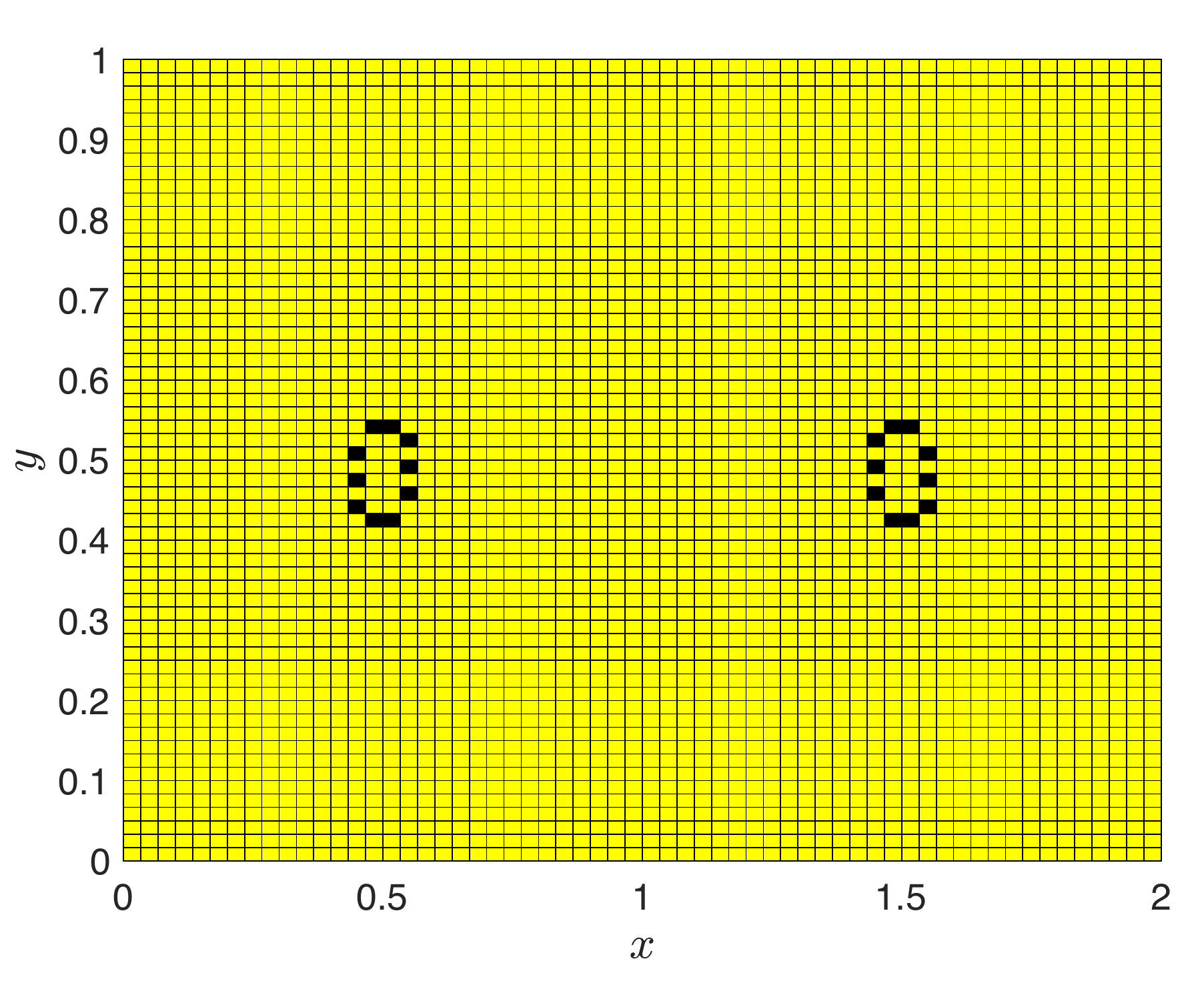}}
\subfigure[]{\includegraphics[scale=.23]{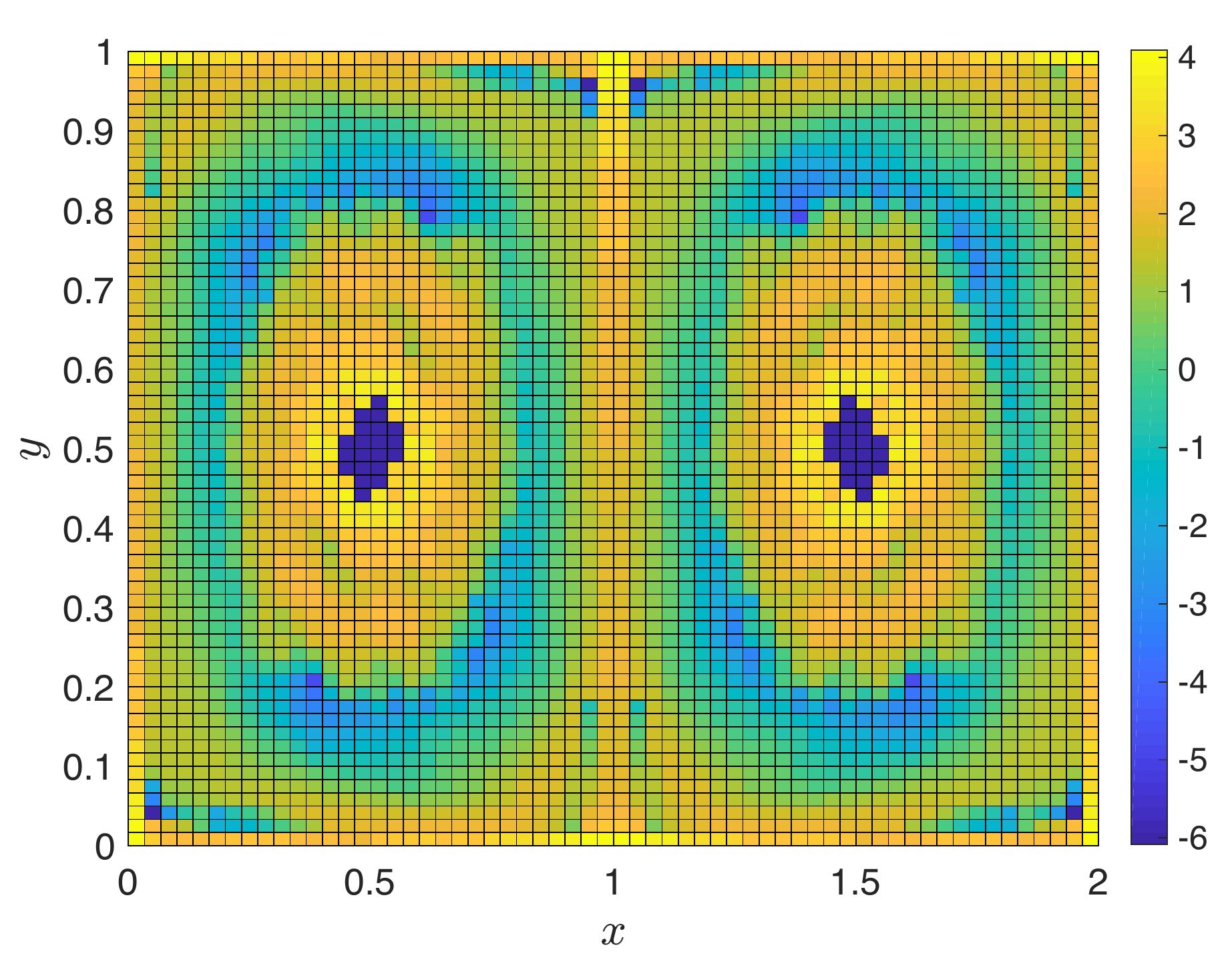}}
\caption{(a) Position of the 20 actuators. (b) Total information transfer from the actuator cells to all the cells in the state space in $\log$ scale.}\label{fig_actuators}
\end{figure}
The position of the actuators, obtained using the optimization problem (\ref{optimization_convex}) is shown in Fig. \ref{fig_actuators}(a) and the information transfer to all the cells from these 20 actuator cells are shown in Fig. \ref{fig_actuators}(b). It should be noted that the information transfer is not non-zero to all the cells from these 20 cells and thus to achieve full coarse controllability, one would require more actuators.

\emph{2. Sensor placement in building systems}
\newline
The vector field used in this example is the average velocity field obtained from a detailed finite element-based simulation of Navier Stokes equation. For the purpose of simulation, we only employ a two dimensional slice of the three dimensional velocity field as shown in Fig. \ref{velocity_field}. The dimensions of the room are as follows: $0\leq x \leq  1.52 m $ and  $0 \leq y \leq 1.68 m$. The order of magnitude
for the velocity field is $O(1)$. The Reynolds number of the flow is $Re = 76 725$ and the Prandtl number $Pr = 0.729$.
\begin{figure}[htp!]
\centering
{\includegraphics[scale=.32]{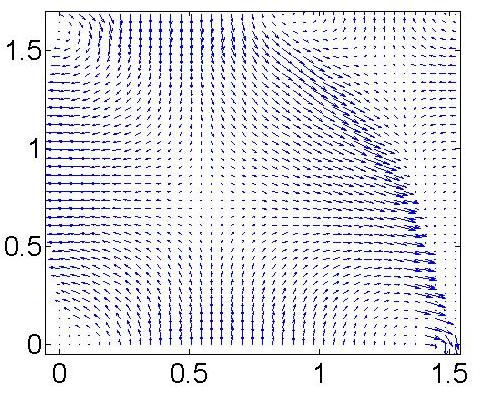}}
\caption{Fluid flow vector field in building.}\label{velocity_field}
\end{figure}

For the finite-dimensional approximation of the P-F operator we divide the state space into $60\times 60$ divisions and choose 6 sensors for maximizing the coarse observable region.

\begin{figure}[htp!]
\centering
\subfigure[]{\includegraphics[scale=.225]{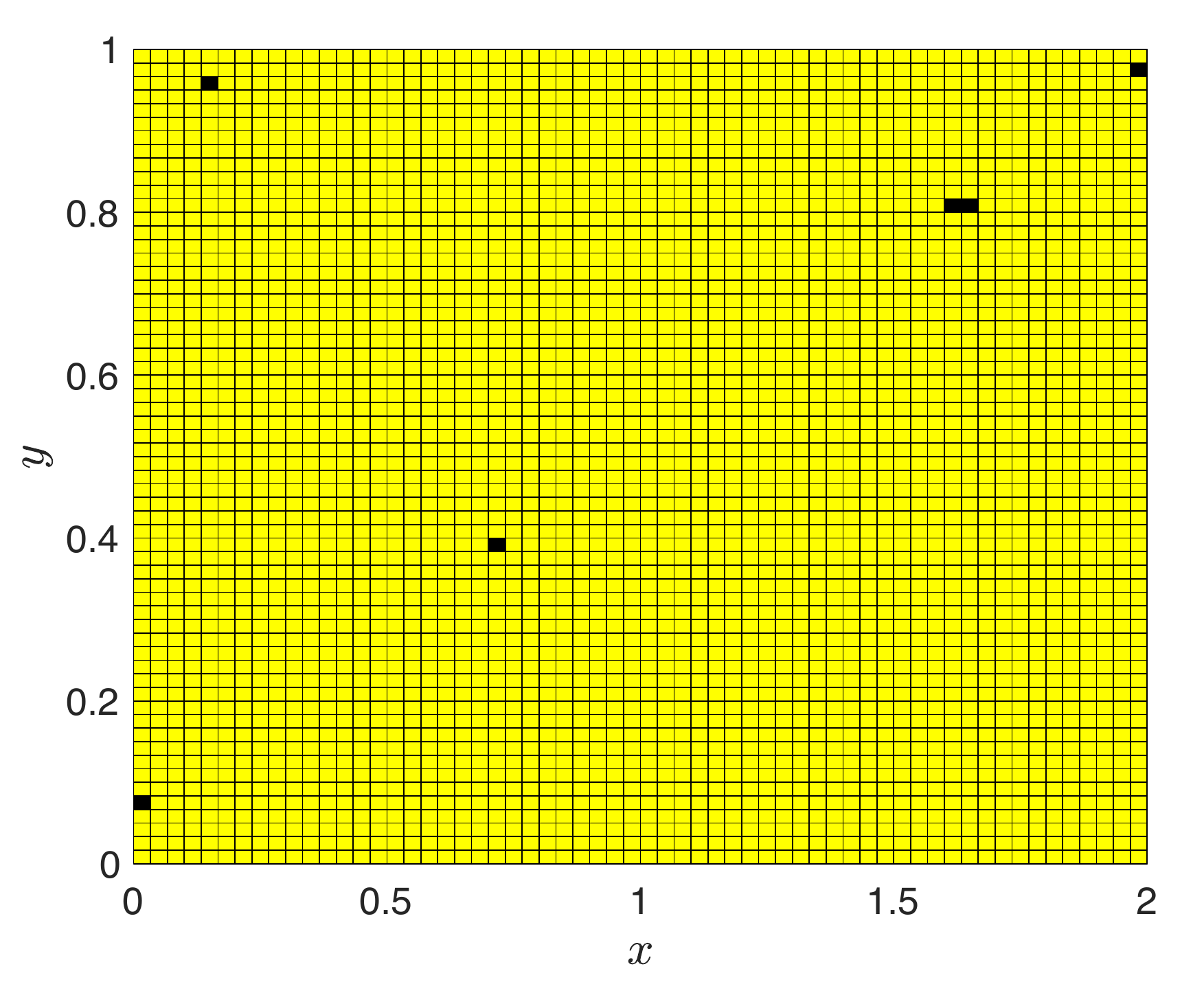}}
\subfigure[]{\includegraphics[scale=.225]{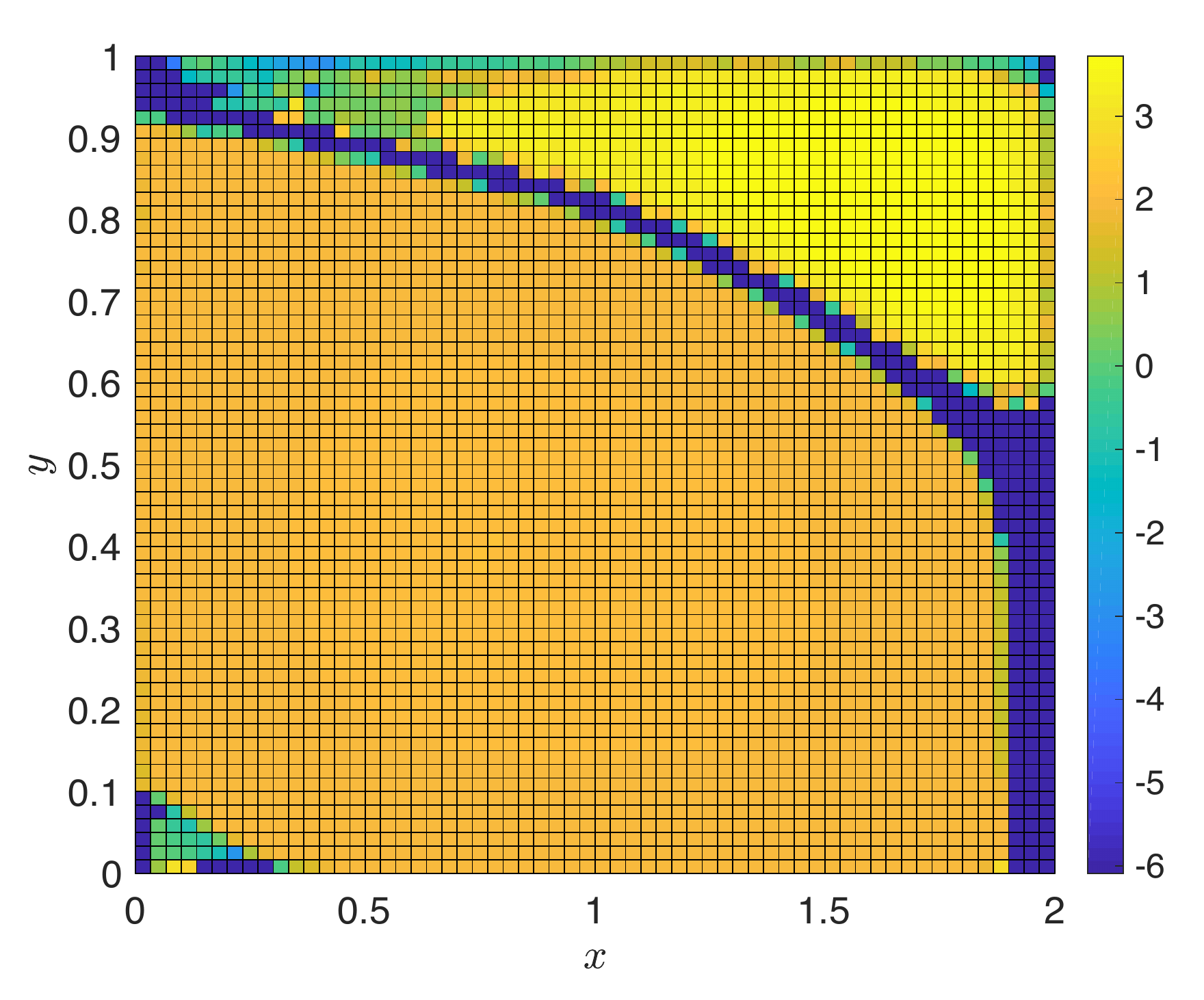}}
\caption{(a) Position of the 6 sensors. (b) Total information transfer to the sensor cells from all the cells in the state space in $\log$ scale.}\label{fig_sensors}
\end{figure}

In Fig. \ref{fig_sensors}(a) the position of the 6 sensors are shown (black boxes). The corresponding coarse observable region is shown in Fig. \ref{fig_sensors}(b) where the information transfer to the sensor cells from all the other cells are plotted in log scale. It is observed that most of the state space is observable with the 6 sensors, with the separatrix not being observed. This is because the separatrix divides the state space into invariant regions and the 6 sensors are placed in each of these invariant subspaces. Hence these sensors can \emph{observe} only the corresponding invariant subspace and cannot observe the separatrix and hence very little information flows from the separatirx to the sensors. This is indicated by the blue region in Fig. \ref{fig_sensors}(b). Ideally there should not be any information flow from the separatrix to the sensors, however in the simulations to to finite partitions, there is very small non-zero information flow from the separatrix to the sensors.

\section{Conclusions}\label{section_conclusion}
In this paper we extend the concept of information transfer between the states in a dynamical system to define information transfer between the measureable sets of a dynamical system and we show the defined information transfer measure satisfy the intuitions of information flow, namely, zero transfer and transfer asymmetry. We further show how information transfer is connected with the dynamical system concepts of ergodicity and mixing and provide necessary and sufficient conditions in terms of information transfer measure for a system to be ergodic or mixing. Finally, we formulate a convex optimization problem in terms of information transfer to address the problem of optimal placement of actuators and sensors for control of non-equilibrium dynamics and demonstrate the efficiency of the developed framework on two different systems.

\bibliographystyle{IEEEtran}
\bibliography{ref,ref1,reference}

\end{document}